\documentclass[preprint,review,12pt]{elsarticle}
\usepackage{graphics}
\usepackage{graphicx}
\usepackage{epsfig}
\usepackage{bm}
\usepackage{subfigure}
\usepackage{amsthm}
\usepackage{arydshln}
\usepackage{lineno}
\usepackage{amssymb}
\usepackage{epsfig}
\usepackage{rotating}
\usepackage[hang,stable]{footmisc}
\usepackage{color}
\usepackage{arydshln}
\usepackage{latexsym}
\usepackage{amsmath}
\usepackage{booktabs}
\usepackage{tikz}
\usepackage{ifpdf}
\usepackage{CJK}
\usepackage{multirow}
\usepackage{morefloats}
\usepackage{afterpage}
\afterpage{\clearpage}

\newtheorem{theorem}{Theorem}
\newtheorem{corollary}[theorem]{Corollary}

\newtheorem{lemma}[theorem]{Lemma}
\newtheorem{definition}{Definition}
\newtheorem{remark}{Remark}

\newtheorem{algorithm}{Algorithm}[section]

\def\be{\begin{equation}}
\def\ee{\end{equation}}
\def\bea{\begin{eqnarray}}
\def\eea{\end{eqnarray}}
\def\bea*{\begin{eqnarray*}}
\def\eea*{\end{eqnarray*}}

\def\bt{\begin{theorem}}
\def\et{\end{theorem}}
\def\bl{\begin{lemma}}
\def\el{\end{lemma}}
\def\br{\begin{remark}}
\def\er{\end{remark}}
\def\bc{\begin{corollary}}
\def\ec{\end{corollary}}
\def\bd{\begin{definition}}
\def\ed{\end{definition}}
\def\bd{\begin{algorithm}}
\def\ed{\end{algorithm}}
\newcounter{test}
\newtheorem*{problem0}{Problem P$_0$}

\newtheorem*{problemP}{Problem P$_0^r$}

\journal{Asia-Pacific Journal of Chemical Engineering}
\begin{document}
\begin{CJK*}{GBK}{kai}
\begin{frontmatter}

\title{Computational Optimal Control of the Saint-Venant PDE Model Using the Time-scaling Technique~\footnote{This work was supported by the National Natural Science Foundation of China
grants (61473253, 61320106009),  the National High Technology Research and Development Program of China 2012AA041701, and  Innovation Joint Research Center for Cyber-Physical-Society System.
 }}

\author{
Tehuan Chen$^a$, Chao Xu$^{a,}$\footnote{Correspondence to: Chao Xu, Email: cxu@zju.edu.cn}}
\address{$^a$ State Key Laboratory of Industrial Control Technology and Institute of Cyber-Systems \& Control, Zhejiang University, Hangzhou, Zhejiang 310027, China.}

\begin{abstract}
This paper proposes a new time-scaling approach for computational optimal control of a distributed parameter system governed by the Saint-Venant PDEs. We propose the time-scaling approach, which can change a uniform time partition to a nonuniform one. We also derive the gradient formulas by using the variational method. Then the method of lines (MOL) is applied to compute the Saint-Venant PDEs after implementing the time-scaling transformation and the associate costate PDEs. Finally, we  compare the optimization results using the proposed time-scaling approach with the one not using it. The simulation result demonstrates the effectiveness of the proposed time-scaling method.
\end{abstract}


\end{frontmatter}

\section{Introduction}
The one dimensional (1D) Saint-Venant (SV) model  is a nonlinear hyperbolic system governed by quasilinear PDEs which can be obtained from the full Navier-Stokes equations (NSE) under certain assumptions and simplifications (i.e., \cite{zhou2004lattice,litrico2009modeling}). In hydraulics, the SV model is widely used to describe transient dam break analysis, open-channel flows and surface runoff. In addition, many phenomena arising in physical applications can be also modeled by the SV model, such as fluid flows in gas distribution pipeline networks, open channel flows, multiphase flow in  pipelines to transport crude oil over long distances (i.e., \cite{zhai2012leak,rao2002contribution}), just to name a few.
In this work, we are interested in a boundary control problem of water hammer phenomenon while manipulating pipeline valves in large scale facilities for liquid distribution. Water hammer is also known as hydraulic shock which is a sharp pressure transition caused by changing the fluid motion state suddenly to halt or a reversed flow direction. This pressure wave could cause harmful effects to the hydraulic facilities, from noise and structural vibration to critical pipe component collapse. There are many applications for mitigation of water hammer, such as oil pipelines leakage \cite{xu2015sensor}, spacecraft propulsion systems \cite{lecourt2007experimental}, and even cardiovascular flow of blood vessels \cite{Pedley1980flow}.
Therefore, passive mitigation methods are widely used to control water hammer, such as accumulators, expansion tanks and surge tanks \cite{ghidaoui2005review}.
The proposed strategy in the current work is to generate  valve actuation command through computational optimization techniques based on the dynamic PDE model of water hammer, which can reduce the hydraulic shock as much as possible. Making boundary valves as active actuation could be an alternative or supplement to various passive protection measures.

Essentially, mitigation of water hammer using boundary valve actuation can be considered as a boundary stabilization problem in terms of the SV model in the point of view of PDE control.  The characteristic method is one of the most important methods in the boundary control of SV model \cite{cen2014boundary,de2003boundary}. There are mainly two streams of approaches of boundary stabilization of hyperbolic PDEs {based on the characteristic method}, including the Lyapunov functional method (e.g., \cite{coron2007strict,dos2008boundary}) and the backstepping technique~\cite{krstic2008boundary}.
A strict Lyapunov function for hyperbolic systems of conservation laws is presented in~\cite{coron2007strict} which can generate a boundary control law to guarantee the local convergence of the state towards a desired set point. The static feedback control law can be implemented as a feedback of the state only measured on the boundaries.
A feedback control strategy is proposed in \cite{dos2008boundary} which ensures that the water level and water flow can converge to the equilibrium exponentially.
The backstepping technique has been extended to handle boundary stabilization of $2 \times 2$ hyperbolic linear and quasilinear PDEs, which allows $L^2$-exponential convergence of the closed-loop and state estimation dynamics \cite{coron2013local,vazquez2011backstepping}. 
Recently, a receding horizon optimal control (RHOC) for water hammer mitigation is investigated for hydraulic pipeline systems described by the linearized SV model \cite{pham2014predictive}. The approximate dynamic programming (ADP) framework is extended to a distributed parameter system described by a set of hyperbolic PDEs \cite{joy2011approximate}.

The current work considers a computational optimal control of the nonlinear SV model in contrast to a feedback stabilizing controller. Running the computational optimal control offline combined with an online tracking controller could be promising to realize a feedback controller for water hammer mitigation in practice. In general, there are mainly two categories of approaches to handle computational optimal control of infinite dimensional systems governed by PDEs, i.e., \textit{discretize-then-optimize} (DTO) \cite{yu2014approximation} and \textit{optimize-then-discretize} (OTD) \cite{fredi2010optimal}. In the framework of DTO, PDEs are first discretized into finite dimensional systems governed by ODEs using various numerical methods, such as the finite volume method (FVM), the lattice Boltzmann method (LBM), and the method of lines (MOL). Then, classical computational techniques can be applied to solve the reduced optimal control problem, such as the control parameterization method, the time-scaling method and the exact penalty method \cite{linsurvey2013,li2013time,Teo1991,liu2014computational}. While in the framework of OTD, optimality conditions and gradient formula can be derived directly based on the PDEs and solve the coupled state and co-state PDEs using various numerical techniques \cite{wang2011application}. 

In this paper, we extend the control parameterization method for finite dimensional control systems to an infinite dimensional system which is governed by the SV model (e.g., \cite{Teo1991}).  We developed a discretize-then-optimize computational approach for solving  optimal control strategy of the SV model in \cite{chent2014}. This approach first uses the  finite-difference method to approximate the PDE model by a system of ODEs, then applies  control parameterization  \cite{Teo1991}  to approximate the boundary control function. While in~\cite{chent2015}, we propose an  alternative computational approach in which control parameterization is applied directly to the original SV model, then finite-difference methods are used  to solve both the  PDE model and costate equations. In both \cite{chent2014} and \cite{chent2015}, the time partition used to parameterize the control input is equally divided. However, we realize that the control trajectory varies slop at different time instance and this motivates us to use less parameters for slowly changing segments but more for comparably fast changing ones. Therefore, we add a new optimization decision variable for the temporal step in control parameterization. This allows us to adaptively select the best switching time instants, which result in a better control approximation. This ideal is called the time-scaling technique in the literature of computational optimal control of finite dimensional systems \cite{Teo1991} but not complete for infinite dimensional systems governed by PDEs.   

The rest of the paper is organized as follows. In Section \ref{Statement}, we state an
optimal control problem for fluid flow during valve closure.  In Section \ref{Time Grid Adaptive}, the control parameterization method of the SV model using the time-scaling  approach is applied to approximate the boundary control by piecewise linear functions. Then, it changes the boundary optimal problem to optimal parameter selection problem. In Section \ref{sensitivity},  we obtain the costate equations together  with their boundary conditions as well as terminal conditions and the gradient formulas are derived  by using the variational analysis method with respect to the control and time parameters. In Section \ref{numerical}, we  use the  MOL to compute the solutions of the state system and its costate system.  Finally, we carry out numerical simulations  to compare the control trajectories when the  time-scaling approach is applied and not, respectively.
\section{Statement of the Optimal Control Problem}\label{Statement}
The mathematical formulation of the optimal control problem with respect to the SV model can be stated as follows:
\begin{equation}
\begin{aligned}
\min_{u}J(u(t)) = \frac{1}{T}\int_0^T {\bigg[\frac{p(L,t) - P}{\bar P}\bigg]^{2\gamma} } dt + \frac{1}{{LT}}{\int_0^L \int_0^T  {\bigg[\frac{p(l,t) - P}{\bar P}\bigg]^{2\gamma}  dtdl}}, \label{Objective funtion}
\end{aligned}
\end{equation}
where  $l\in[0,L]$ denotes the spatial,  $t\in[0,T]$ is the time, $\gamma$ is a positive integer, and $\bar{P}$ is a given constant datum. The objective function (\ref{Objective funtion}) consists of two terms: the first term penalizes pressure fluctuation at the  terminus while the second term penalizes pressure fluctuation at all points over the physical domain.  Considering the actuator situated at terminal point which contains sensitive components that can be easily damaged,  we place special emphasis at this point.  The  pressure drop $p(l,t)$ is the unique solution of the following initial value problem
\begin{subequations}\label{orignalsystem}
\begin{align}
&H_1(l,t)=\frac{{\partial v(l,t)}}{{\partial t}}   + \frac{1}{\rho}\frac{{\partial p(l,t)}}{{\partial l}} + \frac{{f v(l,t) \left| {v(l,t)} \right| }}{{2D}}=0,
\label{system:1}\\
&H_2(l,t)=\frac{{\partial p(l,t)}}{{\partial t}} + {\rho c^2 }\frac{{\partial v(l,t)}}{{\partial l}}=0, \label{system:2}\\
&p(l,0) = \phi_1 (l),~~v(l,0) = \phi_2 (l),~~l\in[0,L], \label{initial conditions}
\end{align}
\end{subequations}
where  $v(l,t)$ is the flow velocity, $\phi_1 (l)$ and $\phi_2(l)$ are given functions describing the initial state of the pipeline, $D$ is the cross-sectional area, $c$ is the wave velocity, $f$ is the Darcy-Weisbach friction factor and $\rho$ is the flow density which is usually considered as a constant. The benchmark model is shown in Figure \ref{pipeline}, where a  pipeline of length $L$ is used to transport fluid from a reservoir to a terminus. Then the boundary conditions for system~(\ref{orignalsystem}) are chosen as
\begin{subequations}\label{bcontrol}
\begin{align}
&p(0,t) =P,\label{boundary condition001}\\
&{v(L,t) = u (t)},~~t\in[0,T],\label{boundary condition002}
\end{align}
\end{subequations}
where $P$ is the  constant pressure generated by the reservoir which is very common in practice. $u(t)$ is a boundary control variable that models actuation such as a valve or water gate at the system terminus and subjected to the following constraints
\begin{subequations}\label{bcont}
\begin{align}
&u(0) = u_{\max}, \label{control initial}\\
&u(T) = 0, \label{control constrain1}
\end{align}
\end{subequations}
where $u_{\max}$ denotes the  maximum  velocity.
\begin{remark}
Note that the open channel flows can be also modeled by the SV model. However, the variables of the flow are flow speed and water lever. This is different from the pressure pipe flow considered in this paper. For more information on open channel flows, please refer to \cite{litrico2009modeling}.
\end{remark}
\begin{figure}
\centering\includegraphics[scale=0.9]{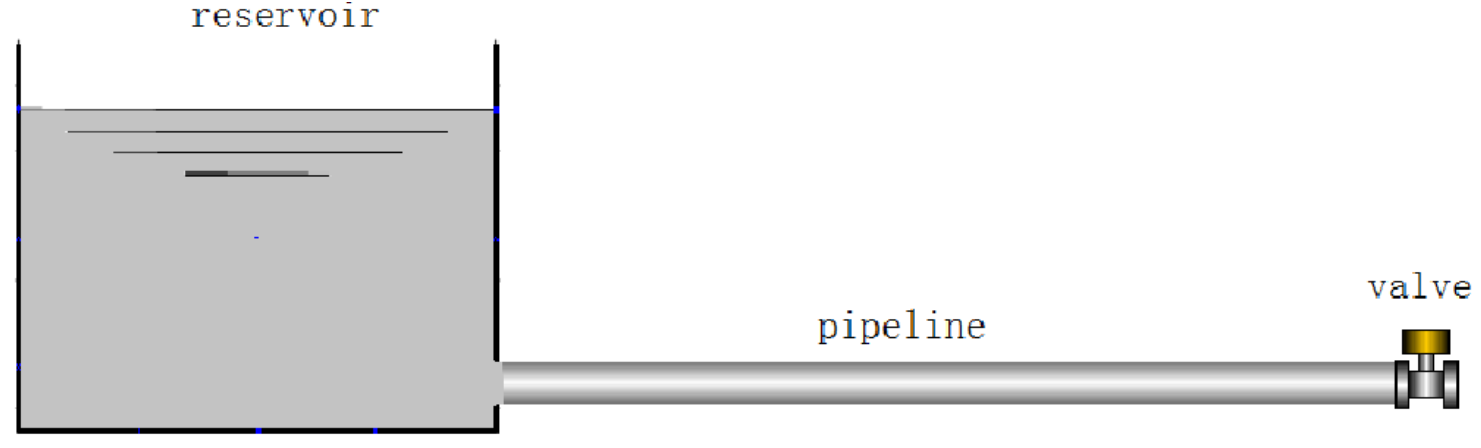}
\caption{General layout of the pipeline system }
\label{pipeline}
\end{figure}

%
%
%
\begin{problem0}
\textit{Given the system (\ref{system:1}) (\ref{system:2}) with initial conditions (\ref{initial conditions}) and  boundary conditions (\ref{bcontrol}), choose the $u(t)$  with initial conditions (\ref{control initial}) to minimize the objective function (\ref{Objective funtion})  subject to the terminal  control constraint (\ref{control constrain1}).}
\end{problem0}

\section{Time-scaling  Approach} \label{Time Grid Adaptive}
By considering the flow rate is continuous, we can approximate the control signal $u(t)$ by piecewise-linear basis functions:
\be\label{piece-linear case}
u(t)\approx  \sum\limits_{k = 1}^r (\sigma_1^kt+ \sigma_2^k)\chi_{[t_{k - 1} ,t_k )}(t),
\ee
where  $\bm \sigma^k \triangleq \{ \sigma^k_1,  \sigma^k_2 \} \in \mathbb{R}^2$, $k=1, \dots, r,$ are parameter vectors to be optimized and $\chi _{[t_{k - 1} ,t_k )}(t) $ is the indicator function defined by
\begin{equation}
\chi_{[t_{k-1},t_k)}(t)=
\begin{cases}
1,&\text{if $t\in [t_{k-1},t_k)$},\\
0,&\text{otherwise},
\end{cases}
\end{equation}
and $t_k, ~k = 0, \dots, r,$ are switching points such that
\begin{equation}
0 = t_0 < t_1 < t_2 < \dots < t_{r-1} < t_r = T.
\end{equation}
Due to the continuity of flow rate, we have
\be\label{piecelinearcase}
\sigma_1^kt_k+ \sigma_2^k=\sigma_1^{k+1}t_{k}+ \sigma_2^{k+1},\quad k=1,...,r-1.
\ee
Furthermore, to ensure that the initial condition (\ref{control initial}) and terminal control constraint (\ref{control constrain1})  is satisfied  (or the compatibility condition), we must have
\begin{equation}\label{terminalconstraint}
\sigma_2^1=u_{\max}, \quad \sigma_1^rT+ \sigma_2^r=0.
\end{equation}
The time-scaling approach is to find the best temporal partition of each interval $[t_{k-1},t_k]$, which means that we consider the switching points as the optimized parameters. However, switching time problem is difficult to solve, so we should transform it into a new problem with fixed switching times \cite{loxton2008optimal}. Thus, the time-scaling  function is defined as follows:
\begin{equation}
t(s)\triangleq \psi (s\left| \bm \theta  \right.)=
\begin{cases}
\sum\limits_{k = 1}^{\left\lfloor s \right\rfloor } {\theta ^k }  + \theta ^{\left\lfloor s \right\rfloor  + 1} (s - \left\lfloor s \right\rfloor ),&\text{if $s\in [0,r)$},\\
T,&\text{$s=r$},
\end{cases}
\end{equation}
where $\left\lfloor s \right\rfloor$ donates an integer which is not larger than $s$. The relationship between $t$ and  $s$ can be also defined through the following differential equation:
\begin{equation}\label{time diferential}
\begin{aligned}
&\frac{{dt(s)}}{{ds}} = \sum\limits_{k = 1}^r {\theta ^k }\chi_{[k-1,k)}(s), \quad s\in [0,r],\\
&t(0)=0,
\end{aligned}
\end{equation}
where $\theta ^k  = t_k  - t_{k - 1}$ and $\theta ^k> 0$.

We change the original time variable ``$t$''  into a new auxiliary variable ``$s$''. Then the approximate piecewise-linear control (\ref{piece-linear case}) can be written as
\begin{equation}\label{approximate piecewise}
u^r(s;\bm \sigma,\bm \theta) =\sum\limits_{k = 1}^r {\bigg\{\sigma _1^k (\sum\limits_{k = 1}^{\left\lfloor s \right\rfloor } {\theta ^k }  + \theta ^{\left\lfloor s \right\rfloor  + 1} (s - \left\lfloor s \right\rfloor )) +\sigma _2^k\bigg \} } \chi_{[{k - 1} ,k )}(s).
\end{equation}
By denoting
\begin{equation}
\begin{aligned}
\tilde p(l,s) = p(l,\psi (s\left|\bm \theta  \right.)),\quad \tilde v(l,s) = v(l,\psi (s\left| \bm \theta  \right.)),\\
\end{aligned}
\end{equation}
the equation (\ref{system:1}) becomes
\begin{equation}\label{transsystem:1}
\begin{aligned}
\dot  {\tilde  v}(l,s) &= \frac{{\partial v(l,\psi (s\left|\bm \theta  \right.)) }}{{\partial s}} = \frac{{\partial  v(l,\psi (s\left| \bm \theta  \right.)) }}{{\partial t}}\frac{{ \partial \psi (s\left| \bm \theta  \right.)}}{{\partial s}}\\
&= \theta^k \bigg[- \frac{1}{\rho}\frac{{\partial \tilde{p}(l,s)}}{{\partial l}} - \frac{{f \tilde{v}(l,s) \left| {\tilde{v}(l,s)} \right| }}{{2D}}\bigg],\quad s \in (k-1,k),\quad k=1,\dots,r,
\end{aligned}
\end{equation}
and the transformed form of (\ref{system:2})  can be obtained following the same procedure in deriving (\ref{transsystem:1}). Then the SV model becomes
\begin{subequations}\label{system001}
\begin{align}
&H_1(l,s)=\frac{{\partial \tilde{v}(l,s)}}{{\partial s}}   + \theta^k \frac{1}{\rho}\frac{{\partial \tilde{p}(l,s)}}{{\partial l}} + \theta^k \frac{{f \tilde{v}(l,s) \left| {\tilde{v}(l,s)} \right| }}{{2D}}=0,
\label{system:01}\\
&H_2(l,s)=\frac{{\partial {\tilde{p}(l,s)}}}{{\partial s}} + \theta^k {\rho c^2 }\frac{{\partial \tilde{v}(l,s)}}{{\partial l}}=0, \label{system:02}
\quad s \in (k-1,k), \quad k=1,\dots,r,\\
&\tilde p(l,0) = p(l,\psi (0\left|\bm \theta  \right.))=\phi_1(l),\quad \tilde v(l,0) = v(l,\psi (0\left|\bm \theta  \right.))=\phi_2(l). \label{system:03}
\end{align}
\end{subequations}
Under the approximation (\ref{approximate piecewise}) for the control input sequence, the objective function (\ref{Objective funtion}) becomes
\begin{equation}
\begin{aligned}
&{J^r(\bm {\sigma},\bm \theta)}\\&= \frac{1}{T}\int_0^T {\bigg[\frac{p^r(L,t) - P}{\bar P}\bigg]^{2\gamma} } dt + \frac{1}{{LT}} \int_0^L \int_0^T {{\bigg[\frac{p^r(l,t) - P}{\bar P}\bigg]^{2\gamma}  dtdl}}
\\&=\frac{1}{T}\sum\limits_{k = 1}^r \int_{t_{k-1}}^{t_k} {\bigg[\frac{p^r(L,t) - P}{\bar P}\bigg]^{2\gamma} } dt + \frac{1}{{LT}}  {\int_0^L \bigg \{ \sum\limits_{k = 1}^r  \int_{t_{k-1}}^{t_k} {\bigg[\frac{p^r(l,t) - P}{\bar P}\bigg]^{2\gamma}  dt \bigg \} dl}}
\\&=\frac{1}{T}\sum\limits_{k = 1}^r \int_{\psi (k-1\left| \bm \theta  \right.)}^{\psi (k\left| \bm \theta  \right.)} {\bigg[\frac{p^r(L,t) - P}{\bar P}\bigg]^{2\gamma} } dt + \frac{1}{{LT}}  {\int_0^L \bigg \{ \sum\limits_{k = 1}^r  \int_{\psi (k-1\left|\bm \theta  \right.)}^{\psi (k\left|\bm \theta  \right.)} {\bigg[\frac{p^r(l,t) - P}{\bar P}\bigg]^{2\gamma}  dt \bigg \} dl}}
\\&=\frac{1}{T}\sum\limits_{k = 1}^r \int_{k-1}^{k} \theta^k {\bigg[\frac{\tilde{p}^r(L,s) - P}{\bar P}\bigg]^{2\gamma} } ds + \frac{1}{{LT}}  {\int_0^L \bigg \{ \sum\limits_{k = 1}^r  \int_{k-1}^{k} \theta^k {\bigg[\frac{\tilde{p}^r(l,s) - P}{\bar P}\bigg]^{2\gamma}  ds \bigg \} dl}},
\label{obj-rtrans}
\end{aligned}
\end{equation}
where $p^r(l,t)$, $\tilde{p}^r(l,s)$ denote the solution of system (\ref{system:1}) (\ref{system:2}) with $u(t)=u^r(t;\bm {\sigma})$ and  system (\ref{system:01}) (\ref{system:02}) with $u(t)=u^r(s;\bm {\sigma},\bm \theta)$, respectively.

Moreover, we have the following linear constraint due to the fixed total time derivation of the valve operation process:
\be\label{linearconstrains}
\theta_1+\theta_2+\cdots+\theta_r=T.
\ee
Then the continuity condition of the flow rate in (\ref{piecelinearcase}) becomes following nonlinear constraints:
\be\label{nonlinearconstrains}
\sigma_1^k\sum\limits_{m = 1}^{k } {\theta ^k }  + \sigma_2^k=\sigma_1^{k+1}\sum\limits_{ m=1}^{k } {\theta ^k }  + \sigma_2^{k+1},\quad k=1,..., r-1.
\ee
\begin{problemP}
\textit{Given the system (\ref{system:01}) (\ref{system:02}) with boundary conditions (\ref{boundary condition001}) (\ref{approximate piecewise}) and initial conditions (\ref{system:03}), choose the $u^r(s;\bm \sigma,\bm \theta)$ to minimize the objective function (\ref{obj-rtrans})  subject to the  constraints  (\ref{terminalconstraint}), (\ref{linearconstrains}), (\ref{nonlinearconstrains}).}
\end{problemP}
\section{Gradient Computation} \label{sensitivity}
Problem P$^r_0$ becomes a nonlinear programming problem. 
Since its gradient is an implicit function, we rewrite the objective function (\ref{obj-rtrans}) and the variational method~\cite{cao2003adjoint,weinstock2012calculus,moura2011optimal} is used to obtain the gradient formulas. The augmented objective function is defined as
\begin{equation}
\begin{aligned}\label{Objective funtion lagerlanri1}
{J^r(\bm {\sigma},\bm \theta)} &= \frac{1}{T}\sum\limits_{k = 1}^r \int_{k-1}^{k} \theta^k {\bigg[\frac{\tilde{p}^r(L,s) - P}{\bar P}\bigg]^{2\gamma} } ds + \frac{1}{{LT}}  \int_0^L \Bigg \{ \sum\limits_{k = 1}^r  \int_{k-1}^{k} \bigg\{ \theta^k  \bigg[\frac{\tilde{p}^r(l,s) - P}{\bar P}\bigg]^{2\gamma} \\& \quad + \tilde{\lambda}(l,s) H_1 (l,s)+ \tilde{\mu} (l,s) H_2 (l,s) \bigg\} ds \Bigg \} dl,
\end{aligned}
\end{equation}
where $\tilde{\lambda}(l,s)$, $\tilde{\mu}(l,s)$ are the Lagrangian multipliers and $H_1(l,s)$ , $H_2(l,s)$ are defined in~(\ref{system001}). Using integration by parts for (\ref{Objective funtion lagerlanri1}), we can rewrite the objective function as
\begin{equation}
\begin{aligned}
&{J^r(\bm {\sigma},\bm \theta)}\\&=  \frac{1}{{LT}}  \int_0^L \Bigg \{ \sum\limits_{k = 1}^r  \int_{k-1}^{k} \theta^k{ {\bigg\{\bigg[\frac{\tilde{p}^r(l,s) - P}{\bar P}\bigg]^{2\gamma} }}- \bigg[\frac{1}{\rho}\tilde{\lambda}_l(l,s) +\frac {\tilde{\mu}_s(l,s)}{\theta ^k } \bigg] {\tilde{p}^r}(l,s)
\\ &\quad + \tilde{\lambda} \bigg[ \frac{{f \left| {\tilde{v}^r(l,s)} \right| }}{{2D}}  -\frac {\tilde{\lambda}_s(l,s)}{\theta ^k } - {\rho}c^2\tilde{\mu}_l(l,s) \bigg]{\tilde{v}^r(l,s)}      \bigg\} ds\Bigg \}dl\\
 &\quad +\frac{1}{T}\sum\limits_{k = 1}^r \int_{k-1}^{k} \theta^k  {\bigg\{\bigg[\frac{\tilde{p}^r(L,s) - P}{\bar P}\bigg]^{2\gamma} }+ \frac{1}{L\rho}\bigg[\tilde{\lambda}(L,s)\tilde{p}^r(L,s)- \tilde{\lambda}(0,s)P\bigg]
\\ &\quad  +\frac{\rho c^2}{L}\bigg[\tilde{\mu}(L,s)u^r(s;\bm \sigma,\bm \theta)-\tilde{\mu}(0,s)\tilde{v}^r(0,s)\bigg]\bigg\}ds
\\ &\quad + \frac{1}{{LT}}\int_0^L\bigg\{\bigg[\tilde{\lambda}(l,r)\tilde{v}^r(l,r)- \tilde{\lambda}(l,0)\phi_1 (l)\bigg]+\bigg[\tilde{\mu}(l,r)\tilde{p}^r(l,r)-\tilde{\mu}(l,0)\phi_1 (l)\bigg] \bigg\}dl.
\end{aligned}
\end{equation}
\begin{theorem}
The gradient formulas of the objective function with respect to the  $\bm \sigma=[(\bm \sigma^1)^\top,\dots,\bm (\bm \sigma^r)^\top ]^\top$ and $\bm \theta=[\theta^1,\dots,\theta^r ]^\top$  are given by
\be\label{g1-dp}
  \nabla_{\sigma_1^k} J( \bm \sigma,\bm \theta)=\frac{\rho c^2}{TL} \int_{{k-1}}^{{k}} \tilde{\mu}(L,s) {{\theta}^k} (\sum\limits_{k = 1}^{\left\lfloor s \right\rfloor } {\theta ^k }  + \theta ^{\left\lfloor s \right\rfloor  + 1} (s - \left\lfloor s \right\rfloor )) ds,\quad k = 1,\dots,r,
 \ee
 \be\label{g0-dp-0}
  \nabla_{\sigma_2^k} J( \bm \sigma,\bm \theta)=\frac{\rho c^2}{TL} \int_{{k-1}}^{{k}} \tilde{\mu}(L,s)\theta^k ds,\quad k = 1,\dots,r,
 \ee
\be\label{g2-dp}
\begin{aligned}
 \nabla_{\theta^k} J( \bm \sigma)&=\frac{1}{{LT}}\int_0^L{\Bigg\{   \int_{{k-1}}^{{k}} {\bigg\{\bigg(\frac{\tilde{p}^r- P}{\bar P}\bigg)^{2\gamma} }}- \frac{1}{\rho}\tilde{\lambda}_l{\tilde{p}^r}+ \bigg( \tilde{\lambda}  \frac{{f  \left| {\tilde{v}^r} \right| }}{{2D}} - {\rho}c^2\tilde{\mu}_l \bigg ){\tilde{v}^r}   \bigg\}  ds \Bigg\}dl
\\ &\quad +\frac{1}{T}  \int_{{k-1}}^{{k}} {\Bigg\{\bigg[\frac{\tilde{p}^r(L,s) - P}{\bar P}\bigg]^{2\gamma} }+ \frac{1}{L\rho}\bigg\{\tilde{\lambda}(L,s)\tilde{p}^r(L,s)- \tilde{\lambda}(0,s)P\bigg\}
\\ &\quad  +\frac{\rho c^2}{L}\tilde{\mu}(L,s){\bigg\{\sigma_1^k \big[\sum\limits_{k = 1}^{\left\lfloor s \right\rfloor } \theta ^k   + \theta ^{\left\lfloor s \right\rfloor  + 1} (s - \left\lfloor s \right\rfloor )\big] +\sigma_2 ^k  }\bigg\} \Bigg\}ds
\\&\quad+\frac{\rho c^2}{TL}\Bigg\{\sum\limits_{m = k+1}^{r} \int_{{m-1}}^{{m}} \tilde{\mu}(L,s) \sigma_1^m  {{\theta}^m} ds+\int_{ k-1}^{k}\tilde{\mu}(L,s) \sigma_1^k  {{\theta}^k}s ds \Bigg\}    ,\\ \quad k = 1,\dots,r,
\end{aligned}
\ee
where $\tilde{\mu}(l,s)$ and  $\tilde{\lambda}(l,s)$ can be solved from the following costate system
\be\label{costate}
\left\{
\begin{aligned}
& \frac{{2\gamma}{\theta}^k}{\bar P}\bigg[\frac {\tilde{p}^r(l,s) - P}{\bar P}\bigg]^{2\gamma-1}- \frac{1}{\rho}{\theta}^k\frac{{\partial \tilde{\lambda}(l,s) }}{{\partial l}} -  \frac{{\partial \tilde{\mu} (l,s)}}{{\partial s}}=0,\\
&{\theta}^k \tilde{\lambda} \frac{{f  \left| {\tilde{v}^r(l,s)} \right| }}{{D} } -\frac{{\partial  \tilde{\lambda}(l,s) }}{{\partial s}}-{\theta}^k{\rho c^2 }\frac{{\partial \tilde{\mu}(l,s) }}{{\partial l}}=0,\\
&\frac{1}{\rho }\tilde{\lambda}(L,s) +\frac{{2\gamma}L}{\bar P}{\bigg[\frac{\tilde{p}^r(L,s) - P}{\bar P}\bigg]^{2\gamma-1}}=0,\\
&\tilde{\mu}(0,s)=0,\tilde{\lambda}(l,r)=\tilde{\mu}(l,r)=0,
\end{aligned}
\quad  s\in[{k-1}, k),\quad k = 1,\dots,r.
\right.
\ee
\end{theorem}

\begin{proof}
By introducing the variational forms $\bm \theta + \epsilon  \tilde{\bm\theta}$,$\bm \sigma + \epsilon  \tilde{\bm\sigma}$, where $\epsilon$ is an arbitrarily positive constant, $\tilde{\bm\theta}=[\tilde{\theta}^1,\dots,\tilde{\theta}^r]^\top$, $\tilde{\bm\sigma}=[( \tilde{\bm \sigma}^1)^\top,\dots,(\tilde{\bm \sigma}^r)^\top]^\top$ are arbitrarily vectors chosen nontrivially, then (\ref{time diferential}) (\ref{approximate piecewise}) change to
\be
\frac{{dt(s;\bm \theta + \epsilon  \tilde{\bm  \theta})}}{{ds}} = \sum\limits_{k = 1}^r ({\theta ^k +\epsilon \tilde{\theta}^k})\chi_{[k-1,k)}(s), \quad s\in [0,r],
\ee
and
\be
\begin{aligned}
& u ^r(s;\bm \sigma + \epsilon  \tilde{\bm  \sigma},\bm \theta + \epsilon  \tilde{\bm  \theta})\\ &= \sum\limits_{k = 1}^r \bigg\{(\sigma_1^k+\epsilon \tilde{\sigma_1}^k) (\sum\limits_{k = 1}^{\left\lfloor s \right\rfloor } {(\theta ^k +\epsilon\tilde{\theta} ^k} )
+ (\theta ^{\left\lfloor s \right\rfloor  + 1}+\epsilon \tilde{\theta} ^{\left\lfloor s \right\rfloor  + 1}) (s - \left\lfloor s \right\rfloor ))
  +\sigma_2 ^k +\epsilon \tilde{\sigma_2} ^k \bigg \}  \chi_{[{k - 1} ,k )}(s).
\end{aligned}
\ee
The corresponding perturbation for $\tilde{p}^r(l,s)$ and $\tilde{v}^r(l,s)$ are approximated as
\be
\begin{aligned}
\tilde{p}^r(l,s;\bm \theta + \epsilon  \tilde{\bm\theta},\bm \sigma + \epsilon  \tilde{\bm\sigma}) = \tilde{p}^r(l,s;\bm\theta,\bm \sigma) +\sum\limits_{k = 1}^r  \langle \nabla_{\theta^k} \tilde{p}^r(l,s;\bm\theta,\bm \sigma), \epsilon \tilde{\theta}^k\rangle \chi_{[{k - 1} ,k )}(s)
\\\quad +\sum\limits_{k = 1}^r  \langle \nabla_{\sigma^k} \tilde{p}^r(l,s;\bm\theta,\bm \sigma), \epsilon \tilde{\bm\sigma}^k\rangle \chi_{[{k - 1} ,k )}(s)+\mathcal O(\epsilon^2),
\end{aligned}
\ee
\be
\begin{aligned}
\tilde{v}^r(l,s;\bm \theta + \epsilon  \bm\theta,\bm \sigma + \epsilon  \bm\sigma) = \tilde{v}^r(l,s;\bm\theta,\bm \sigma) +\sum\limits_{k = 1}^r  \langle \nabla_{\theta^k} \tilde{v}^r(l,s; \bm\theta,\bm \sigma), \epsilon \tilde{\theta}^k\rangle \chi_{[{k - 1} ,k )}(s)
\\\quad +\sum\limits_{k = 1}^r  \langle \nabla_{\sigma^k} \tilde{v}^r(l,s;\bm\theta,\bm \sigma), \epsilon \tilde{\bm \sigma}^k\rangle \chi_{[{k - 1} ,k )}(s)+\mathcal O(\epsilon^2),
\end{aligned}
\ee
where $\mathcal O(\epsilon^2)$ denotes higher order terms such that $\mathcal O(\epsilon^2)\rightarrow 0$ as $\epsilon\rightarrow 0$, defining the new notation $\eta_1^k=\langle \nabla_{\theta^k} \tilde{p}^r(l,s;\bm\theta,\bm \sigma), \tilde{\theta}^k\rangle,  \eta_2^k=\langle \nabla_{\sigma^k} \tilde{p}^r(l,s;\bm\theta,\bm \sigma), \tilde{\bm \sigma}^k\rangle$ and $\omega_1^k=\langle \nabla_{\theta^k} \tilde{v}^r(l,s; \bm\theta,\bm \sigma), \tilde{\theta}^k\rangle,\omega_2^k=\langle \nabla_{\sigma^k} \tilde{v}^r(l,s;\bm\theta,\bm \sigma), \tilde{\bm \sigma}^k\rangle$.  Then, the perturbed augmented objective function takes the following form
\begin{equation}
\begin{aligned}
&J(\bm \theta + \epsilon  \tilde{\bm\theta},\bm \sigma + \epsilon  \tilde{\bm\sigma})
\\&=   \frac{1}{{LT}}\int_0^L{\Bigg\{ \sum\limits_{k = 1}^r  \int_{{k-1}}^{{k}}({\theta^k +\epsilon \tilde{\theta}^k}) {\bigg\{\bigg[\frac{\tilde{p}^r+\epsilon\eta_1^k+\epsilon\eta_2^k - P}{\bar P}\bigg]^{2\gamma} }}- \bigg[ \frac{1}{\rho}\tilde{\lambda}_l + \frac {\tilde{\mu}_s}{({\theta ^k +\epsilon \tilde{\theta}^k})} \bigg ]{(\tilde{p}^r+\epsilon\eta_1^k+\epsilon\eta_2^k)}
\\&\quad + \bigg[ \tilde{\lambda}  \frac{{f  \left| {(\tilde{v}^r+\epsilon\omega_1^k+\epsilon\omega_2^k)} \right| }}{{2D}} -\frac{ \tilde{\lambda}_s}{\theta ^k +\epsilon \tilde{\theta}^k} - {\rho}c^2\tilde{\mu}_l \bigg ]{(\tilde{v}^r+\epsilon\omega_1^k+\epsilon\omega_2^k)}   \bigg\}  ds \Bigg\}dl\\
 &\quad +\frac{1}{T}\sum\limits_{k = 1}^r  \int_{{k-1}}^{{k}} ({\theta ^k +\epsilon \tilde{\theta}^k}) {\Bigg\{\bigg[\frac{\tilde{p}^r(L,s)+\epsilon\eta_1^k(L,s)+\epsilon\eta_2^k(L,s) - P}{\bar P}\bigg]^{2\gamma} }
 \\ &\quad + \frac{1}{L\rho}\bigg\{\tilde{\lambda}(L,s)\bigg[\tilde{p}^r(L,s)+\epsilon\eta_1^k(L,s)+\epsilon\eta_2^k(L,s)\bigg]- \tilde{\lambda}(0,s)P\bigg\}
\\ &\quad  +\frac{\rho c^2}{L}\Bigg\{\tilde{\mu}(L,s){\bigg\{(\sigma_1^k+\epsilon \tilde{\sigma_1}^k) (\sum\limits_{k = 1}^{\left\lfloor s \right\rfloor } {(\theta ^k +\epsilon\tilde{\theta} ^k} ) + (\theta ^{\left\lfloor s \right\rfloor  + 1}+\epsilon \tilde{\theta} ^{\left\lfloor s \right\rfloor  + 1}) (s - \left\lfloor s \right\rfloor )) +\sigma_2 ^k +\epsilon \tilde{\sigma_2} ^k \bigg \} }
\\&\quad-\tilde{\mu}(0,s)\bigg[\tilde{v}^r(0,s)+\epsilon\omega_1^k(0,s)+\epsilon\omega_2^k(0,s)\bigg]\bigg\}\Bigg\}ds
 \\ &\quad+ \frac{1}{{LT}}\int_0^L\Bigg\{\bigg\{\tilde{\lambda}(l,r)\bigg[\tilde{v}^r(l,r)+\epsilon\omega_1^k(l,r)+\epsilon\omega_2^k(l,r)\bigg]
- \tilde{\lambda}(l,0)\phi_1 (l)\bigg\}
 \\ &\quad+\bigg\{\tilde{\mu}(l,r)\bigg[\tilde{p}^r(l,r) + \epsilon\eta_1^k(l,r)+ \epsilon\eta_2^k(l,r)\bigg]-\tilde{\mu}(l,0)\phi_1 (l)\bigg\} \Bigg\}dl. \label{Objective funtion discrete}
\end{aligned}
\end{equation}
By computing the derivative of $J(\bm \theta + \epsilon  \tilde{\bm\theta},\bm \sigma + \epsilon  \tilde{\bm\sigma})$ with respect to the parameters $\epsilon$, we can obtain
\begin{equation}
\begin{aligned}
&\frac{d{J(\bm \theta + \epsilon  \tilde{\bm\theta},\bm \sigma + \epsilon  \tilde{\bm\sigma})}}{d\epsilon}
\\&=   \frac{1}{{LT}}\int_0^L{\Bigg\{ \sum\limits_{k = 1}^r  \int_{{k-1}}^{{k}} \tilde{\theta}^k {\bigg\{\bigg[\frac{\tilde{p}^r+\epsilon\eta_1^k+\epsilon\eta_2^k - P}{\bar P}\bigg]^{2\gamma} }}-\frac{1}{\rho}\tilde{\lambda}_l {(\tilde{p}^r+\epsilon\eta_1^k+\epsilon\eta_2^k)}
\\&\quad + \bigg[ \tilde{\lambda}  \frac{{f  \left| {(\tilde{v}^r+\epsilon\omega_1^k+\epsilon\omega_2^k)} \right| }}{{2D}} - {\rho}c^2\tilde{\mu}_l \bigg ]{(\tilde{v}^r+\epsilon\omega_1^k+\epsilon\omega_2^k)}   \bigg\}  ds \Bigg\}dl
\\& \quad +  \frac{1}{{LT}}\int_0^L{\Bigg\{ \sum\limits_{k = 1}^r  \int_{{k-1}}^{{k}} {\bigg\{\frac{2\gamma}{\bar P}({\theta ^k +\epsilon \tilde{\theta}^k})\bigg[\frac{\tilde{p}^r+\epsilon\eta_1^k+\epsilon\eta_2^k - P}{\bar P}\bigg]^{2\gamma-1} }}- \bigg[ \frac{1}{\rho}({\theta ^k +\epsilon \tilde{\theta}^k})\tilde{\lambda}_l + {\tilde{\mu}_s} \bigg ]\bigg\}{(\eta_1^k+\eta_2^k)}
\\&\quad +\bigg[ ({\theta ^k +\epsilon \tilde{\theta}^k}) \tilde{\lambda}  \frac{{f  \left| {(\tilde{v}^r+\epsilon\omega_1^k+\epsilon\omega_2^k)} \right| }}{{D}} -{ \tilde{\lambda}_s} - ({\theta ^k +\epsilon \tilde{\theta}^k}){\rho}c^2\mu_l \bigg ]{(\omega_1^k+\omega_2^k)}   \bigg\}  ds \Bigg\}dl\\
 &\quad +\frac{1}{T}\sum\limits_{k = 1}^r  \int_{{k-1}}^{{k}} {\tilde{\theta}^k} {\Bigg\{\bigg[\frac{\tilde{p}^r(L,s)+\epsilon\eta_1^k(L,s)+\epsilon\eta_2^k(L,s) - P}{\bar P}\bigg]^{2\gamma} }
 \\ &\quad + \frac{1}{L\rho}\bigg\{\tilde{\lambda}(L,s)\bigg[\tilde{p}^r(L,s)+\epsilon\eta_1^k(L,s)+\epsilon\eta_2^k(L,s)\bigg]- \tilde{\lambda}(0,s)P\bigg\}
\\ &\quad  +\frac{\rho c^2}{L}\bigg\{\tilde{\mu}(L,s){\bigg\{(\sigma_1^k+\epsilon \tilde{\sigma_1}^k) (\sum\limits_{k = 1}^{\left\lfloor s \right\rfloor } {(\theta ^k +\epsilon\tilde{\theta} ^k} ) + (\theta ^{\left\lfloor s \right\rfloor  + 1}+\epsilon \tilde{\theta} ^{\left\lfloor s \right\rfloor  + 1}) (s - \left\lfloor s \right\rfloor )) +\sigma_2 ^k +\epsilon \tilde{\sigma_2} ^k \bigg \} }
\\&\quad-\tilde{\mu}(0,s)\bigg[\tilde{v}^r(0,s)+\epsilon\omega_1^k(0,s)+\epsilon\omega_2^k(0,s)\bigg]\bigg\}\Bigg\}ds\\
 &\quad +\frac{1}{T}\sum\limits_{k = 1}^r  \int_{{k-1}}^{{k}} ({\theta ^k +\epsilon \tilde{\theta}^k}) \bigg\{\frac{2\gamma}{\bar P}\bigg[\frac{\tilde{p}^r(L,s)+\epsilon\eta_1^k(L,s)+\epsilon\eta_2^k(L,s) - P}{\bar P}\bigg]^{2\gamma-1}
  \\ &\quad + \frac{1}{L\rho}\tilde{\lambda}(L,s) \bigg\}\bigg[\eta_1^k(L,s)+\eta_2^k(L,s)\bigg]
\\ &\quad  +({\theta ^k +\epsilon \tilde{\theta}^k})\frac{\rho c^2}{L}\Bigg\{\tilde{\mu}(L,s)\bigg\{\tilde{\sigma_1}^k (\sum\limits_{k = 1}^{\left\lfloor s \right\rfloor } {(\theta ^k +\epsilon\tilde{\theta} ^k} ) + (\theta ^{\left\lfloor s \right\rfloor  + 1}+\epsilon \tilde{\theta} ^{\left\lfloor s \right\rfloor  + 1}) (s - \left\lfloor s \right\rfloor ))
 \\ &\quad +(\sigma_1^k+\epsilon \tilde{\sigma_1}^k) (\sum\limits_{k = 1}^{\left\lfloor s \right\rfloor } {\tilde{\theta} ^k}  +  \tilde{\theta} ^{\left\lfloor s \right\rfloor  + 1} (s - \left\lfloor s \right\rfloor ))+ \tilde{\sigma_2} ^k \bigg \} -\tilde{\mu}(0,s)\bigg[\omega_1^k(0,s)+\omega_2^k(0,s)\bigg]\Bigg\}dt
 \\ &\quad+ \frac{1}{{LT}}\int_0^L\Bigg\{\tilde{\lambda}(l,r)\bigg[\omega_1^k(l,r)+\omega_2^k(l,r)\bigg] +\tilde{\mu}(l,r)\bigg[\eta_1^k(l,r)+\eta_2^k(l,r)\bigg] \Bigg\}dl.
\end{aligned}
\end{equation}
By substituting $\epsilon=0$, we can obtain
\begin{equation}
\begin{aligned}
&\left. {\frac{d{J(\bm \theta + \epsilon  \tilde{\bm\theta},\bm \sigma + \epsilon  \tilde{\bm\sigma})}}{d\epsilon}} \right|_{\epsilon = 0}
\\&=   \frac{1}{{LT}}\int_0^L{\Bigg\{ \sum\limits_{k = 1}^r  \int_{{k-1}}^{{k}} \tilde{\theta}^k {\bigg\{\bigg[\frac{\tilde{p}^r- P}{\bar P}\bigg]^{2\gamma} }}- \frac{1}{\rho}\tilde{\lambda}_l{\tilde{p}^r}+ \bigg[ \tilde{\lambda}  \frac{{f  \left| {\tilde{v}^r} \right| }}{{2D}} - {\rho}c^2\tilde{\mu}_l \bigg ]{\tilde{v}^r}   \bigg\}  ds \Bigg\}dl
\\&+  \frac{1}{{LT}}\int_0^L{\Bigg\{ \sum\limits_{k = 1}^r  \int_{{k-1}}^{{k}} {\bigg\{\frac{2\gamma}{\bar P}{\theta ^k}\bigg[\frac{\tilde{p}^r - P}{\bar P}\bigg]^{2\gamma-1} }}- \bigg[ \frac{1}{\rho}\theta ^k \tilde{\lambda}_l + {\tilde{\mu}_s} \bigg ]\bigg\}{(\eta_1^k+\eta_2^k)}
\\&\quad +\bigg[ {\theta ^k } \tilde{\lambda}  \frac{{f  \left| {\tilde{v}^r} \right| }}{{D}} -{ \tilde{\lambda}_s} - \theta ^k {\rho}c^2\mu_l \bigg ]{(\omega_1^k+\omega_2^k)}   \bigg\}  ds \Bigg\}dl\\
 &\quad +\frac{1}{T}\sum\limits_{k = 1}^r  \int_{{k-1}}^{{k}} {\tilde{\theta}^k} {\Bigg\{\bigg[\frac{\tilde{p}^r(L,s) - P}{\bar P}\bigg]^{2\gamma} }+ \frac{1}{L\rho}\bigg\{\tilde{\lambda}(L,s)\tilde{p}^r(L,s)- \tilde{\lambda}(0,s)P\bigg\}
\\ &\quad  +\frac{\rho c^2}{L}\bigg\{\tilde{\mu}(L,s){\bigg\{\sigma_1^k (\sum\limits_{k = 1}^{\left\lfloor s \right\rfloor } \theta ^k   + \theta ^{\left\lfloor s \right\rfloor  + 1} (s - \left\lfloor s \right\rfloor )) +\sigma ^k_2  \bigg \} } -\tilde{\mu}(0,s)\tilde{v}(0,s)\bigg\}\Bigg\}dt\\
 &\quad +\frac{1}{T}\sum\limits_{k = 1}^r  \int_{{k-1}}^{{k}} \theta ^k  {\bigg\{\frac{2\gamma}{\bar P}\bigg[\frac{\tilde{p}^r(L,s) - P}{\bar P}\bigg]^{2\gamma-1} + \frac{1}{L\rho}\tilde{\lambda}(L,s) }\bigg\}\bigg[\eta_1^k(L,s)+\eta_2^k(L,s)\bigg]
\\ &\quad  +{{\theta}^k} \frac{\rho c^2}{L}\Bigg\{\tilde{\mu}(L,s)\bigg\{\tilde{\sigma_1}^k (\sum\limits_{k = 1}^{\left\lfloor s \right\rfloor } {\theta ^k }  + \theta ^{\left\lfloor s \right\rfloor  + 1} (s - \left\lfloor s \right\rfloor ))
 \\ &\quad +\sigma_1^k (\sum\limits_{k = 1}^{\left\lfloor s \right\rfloor } {\tilde{\theta} ^k}  +  \tilde{\theta} ^{\left\lfloor s \right\rfloor  + 1} (s - \left\lfloor s \right\rfloor ))+ \tilde{\sigma_2} ^k \bigg \} -\tilde{\mu}(0,s)\bigg[\omega_1^k(0,s)+\omega_2^k(0,s)\bigg]\Bigg\}ds
 \\ &\quad+ \frac{1}{{LT}}\int_0^L\Bigg\{\tilde{\lambda}(l,r)\bigg[\omega_1^k(l,r)+\omega_2^k(l,r)\bigg] +\tilde{\mu}(l,r)\bigg[\eta_1^k(l,r)+\eta_2^k(l,r)\bigg] \Bigg\}dl.
\end{aligned}\label{Objective funtion discrete001}
\end{equation}
The optimality condition to minimize objective function is to force $\delta J(u(t))$ to be zero. By using the fundamental lemma in the calculus of variation \cite{{weinstock2012calculus}}, one can obtain the costate system from (\ref{Objective funtion discrete001}) due to the arbitrary choice of $\tilde{\theta}$ and $\tilde{\sigma}$ in the variational form,
\be \label{differential equation}
\left\{ \begin{aligned} \frac{{2\gamma}{\theta}^k}{\bar P}\bigg[\frac {\tilde{p}^r(l,s) - P}{\bar P}\bigg]^{2\gamma-1}- \frac{1}{\rho}{\theta}^k\frac{{\partial \tilde{\lambda}(l,s) }}{{\partial l}} -  \frac{{\partial \tilde{\mu} (l,s)}}{{\partial s}}=0,\\
{\theta}^k \tilde{\lambda}(l,s) \frac{{f  \left| {\tilde{v}^r(l,s)} \right| }}{{D} } -\frac{{\partial  \tilde{\lambda}(l,s) }}{{\partial s}}-{\theta}^k{\rho c^2 }\frac{{\partial \tilde{\mu}(l,s) }}{{\partial l}}=0,
\end{aligned}
\quad  s\in[{k-1}, k),\quad k = 1,\dots,r,
\right.
\ee
where boundary conditions are
\be \label{condition1}
\left\{\begin{aligned}
\frac{1}{\rho }\tilde{\lambda}(L,t) +\frac{{2\gamma}L}{\bar P}{\bigg[\frac{\tilde{p}^r(L,t) - P}{\bar P}\bigg]^{2\gamma-1}}=0,\\
\tilde{\mu}(0,s)=0,
\end{aligned}
\quad  s\in[{k-1}, k),\quad k = 1,\dots,r.
\right.
\ee
The terminal time conditions at $s=r$ are
\be \label{condition2}
\tilde{\lambda}(l,r)=\tilde{\mu}(l,r)=0.
\ee
By substituting (\ref{differential equation})-(\ref{condition2})   to   (\ref{Objective funtion discrete001}), we can obtain
\begin{equation}
\begin{aligned}
&\left. {\frac{d{J(\bm \theta + \epsilon  \tilde{\bm\theta},\bm \sigma + \epsilon  \tilde{\bm\sigma})}}{d\epsilon}} \right|_{\epsilon = 0}
\\&=   \frac{1}{{LT}}\int_0^L{\Bigg\{ \sum\limits_{k = 1}^r  \int_{{k-1}}^{{k}} \tilde{\theta}^k {\bigg\{\bigg[\frac{\tilde{p}^r- P}{\bar P}\bigg]^{2\gamma} }}- \frac{1}{\rho}\tilde{\lambda}_l{\tilde{p}^r}+ \bigg[ \tilde{\lambda}  \frac{{f  \left| {\tilde{v}^r} \right| }}{{2D}} - {\rho}c^2\tilde{\mu}_l \bigg ]{\tilde{v}^r}   \bigg\}  ds \Bigg\}dl
\\ &\quad +\frac{1}{T}\sum\limits_{k = 1}^r  \int_{{k-1}}^{{k}} {\tilde{\theta}^k} {\Bigg\{\bigg[\frac{\tilde{p}^r(L,s) - P}{\bar P}\bigg]^{2\gamma} }+ \frac{1}{L\rho}\bigg\{\tilde{\lambda}(L,s)\tilde{p}^r(L,s)- \tilde{\lambda}(0,s)P\bigg\}
\\ &\quad  +\frac{\rho c^2}{L}\tilde{\mu}(L,s){\bigg\{\sigma_1^k (\sum\limits_{k = 1}^{\left\lfloor s \right\rfloor } \theta ^k   + \theta ^{\left\lfloor s \right\rfloor  + 1} (s - \left\lfloor s \right\rfloor )) +\sigma_2 ^k  \bigg \} } \Bigg\}ds
\\&\quad+\frac{\rho c^2}{TL} \int_{{k-1}}^{{k}}{{\theta}^k} \tilde{\mu}(L,s)\Bigg\{\tilde{\sigma_1}^k (\sum\limits_{k = 1}^{\left\lfloor s \right\rfloor } {\theta ^k }  + \theta ^{\left\lfloor s \right\rfloor  + 1} (s - \left\lfloor s \right\rfloor )) +\sigma_1^k (\sum\limits_{k = 1}^{\left\lfloor s \right\rfloor } {\tilde{\theta} ^k}  +  \tilde{\theta} ^{\left\lfloor s \right\rfloor  + 1} (s - \left\lfloor s \right\rfloor ))+ \tilde{\sigma_2} ^k \Bigg \}  ds.
\end{aligned}
\end{equation}
Therefore, we can obtain the following gradient formulas with respect to the optimization decision variable,
\be
  \nabla_{\sigma_1^k} J( \bm \sigma)=\frac{\rho c^2}{TL} \int_{{k-1}}^{{k}} \tilde{\mu}(L,s) {{\theta}^k} (\sum\limits_{k = 1}^{\left\lfloor s \right\rfloor } {\theta ^k }  + \theta ^{\left\lfloor s \right\rfloor  + 1} (s - \left\lfloor s \right\rfloor )) ds,\quad k = 1,\dots,r,
 \ee

\be
  \nabla_{\sigma_2^k} J( \bm \sigma)=\frac{\rho c^2}{TL} \int_{{k-1}}^{{k}} \tilde{\mu}(L,s)\theta^k ds,\quad k = 1,\dots,r,
 \ee

\be
\begin{aligned}
 \nabla_{\theta^k} J( \bm \sigma)&=\frac{1}{{LT}}\int_0^L{\Bigg\{   \int_{{k-1}}^{{k}} {\bigg\{\bigg[\frac{\tilde{p}^r- P}{\bar P}\bigg]^{2\gamma} }}- \frac{1}{\rho}\tilde{\lambda}_l{\tilde{p}^r}+ \bigg[ \tilde{\lambda}  \frac{{f  \left| {\tilde{v}^r} \right| }}{{2D}} - {\rho}c^2\tilde{\mu}_l \bigg ]{\tilde{v}^r}   \bigg\}  ds \Bigg\}dl
\\ &\quad +\frac{1}{T}  \int_{{k-1}}^{{k}} {\Bigg\{\bigg[\frac{\tilde{p}^r(L,s) - P}{\bar P}\bigg]^{2\gamma} }+ \frac{1}{L\rho}\bigg\{\tilde{\lambda}(L,s)\tilde{p}^r(L,s)- \tilde{\lambda}(0,s)P\bigg\}
\\ &\quad  +\frac{\rho c^2}{L}\tilde{\mu}(L,s){\bigg\{\sigma_1^k (\sum\limits_{k = 1}^{\left\lfloor s \right\rfloor } \theta ^k   + \theta ^{\left\lfloor s \right\rfloor  + 1} (s - \left\lfloor s \right\rfloor )) +\sigma_2 ^k  }\bigg\} \Bigg\}ds
\\&\quad+\frac{\rho c^2}{TL}\Bigg\{\sum\limits_{m = k+1}^{r} \int_{{m-1}}^{{m}} \tilde{\mu}(L,s) \sigma_1^m  {{\theta}^m} ds+\int_{ k-1}^{k}\tilde{\mu}(L,s) \sigma_1^k  {{\theta}^k}s ds \Bigg\}    ,\\ \quad k = 1,\dots,r.
\end{aligned}
\ee
\end{proof}

\section{Numerical Approximation} \label{numerical}
\subsection {Simulation of the State System}
Using the  method of lines, which has been applied to obtain the numerical solution of  the nonlinear SV  model \cite{balogun1988automatic,georges1994nonlinear},   we can decompose the space domain into equally partitions $L_i  = \left[ {l_{i-1}  ,l_{i} } \right]$, $i = 1,\dots ,N$, where $N$ is an even integer with $l_0  = 0$ and $l_{N}  = L$.  Let $\tilde{v}^r_i (s) = \tilde{v}^r(l_i,s)$, $i = 0,\dots,N,$ and  $\tilde{p}^r_i (s) = \tilde{p}^r(l_i ,s),i = 0,\dots,N$.  We  make the following finite difference approximation scheme
\begin{subequations}
\begin{align}
\hspace{-3mm}\frac{{\partial \tilde{p}_i^r(s)}}{{\partial l}}\hspace{-1mm} &= \hspace{-1mm}\frac{{{\tilde{p}^r_{i+1}}(s) - {\tilde{p}^r_{i }}(s)}}{\Delta l}, \quad i = 0,\ldots, N-1,\label{aproximate1}\\
\hspace{-3mm}\frac{{\partial \tilde{v}_i^r(s)}}{{\partial l}} \hspace{-1mm} &= \hspace{-1mm} \frac{{{\tilde{v}^r_i}(s) - {\tilde{v}^r_{i - 1}}(s)}}{\Delta l} ,\quad i = 1,\ldots, N,\label{aproximate2}
\end{align}
\end{subequations}
where $\Delta l = L/N$. Then,  we  substitute the  approximations (\ref{aproximate1}) and (\ref{aproximate2}) into the transformed dynamic system (\ref{system:01}) and  (\ref{system:02}) to obtain the following finite dimensional representation
\begin{subequations}\label{subequation001}
\begin{align}
&{\dot {\tilde{v}}^r_{i }(s)}  = \theta^k \frac{1}{\rho \Delta l}(\tilde{p}^r_{i } (s) - \tilde{p}^r_{i+1} (s))- \theta^k \frac{{f \tilde{v}^r_i(s) \left| {\tilde{v}^r_i(s)} \right| }}{{2D}}, \quad i = 0,\dots,N-1 , \label{pequation1}\\
&{\dot {\tilde{p}}^r_{i }(s)}  = \theta^k \frac{\rho c^2}{ \Delta l }({\tilde{v}^r_{i-1} (s) - \tilde{v}^r_{i } (s)}), \quad i = 1,\dots,N . \label{pequation2}
\end{align}
\end{subequations}
For the initial conditions, we obtain
\begin{equation}
\tilde p^r(l,0) = \phi_1(l_i),\quad \tilde v^r(l,0) = \phi_2(l_i),\quad i = 0,\dots,N. \label{initial conditions2}
\end{equation}
For the boundary conditions, we have
\begin{equation}
\tilde{p}^r_0(s) =P,~~\tilde{v}^r_{N}(s) = u^r(s;\bm \sigma,\bm \theta).\label{boundary condition2}
\end{equation}
Combining the transformed dynamic system (\ref{subequation001}) with  the initial conditions (\ref{initial conditions2}) and the boundary conditions (\ref{boundary condition2}), we can numerically solve $\tilde{v}^r(l,s)$ and $\tilde{p}^r(l,s)$ forward in time.
\subsection{Numerical Discretization of the Costate System}
Similarly, the method of lines is also applied to solve the costate system (\ref{costate}) numerically. Let $\tilde{\lambda}_i (s) = \tilde{\lambda}(l_i,s)$, $i = 0,\dots,N,$ and  $\tilde{\mu}_i (s) = \tilde{\mu}(l_i ,s),i = 0,\dots,N$, and we can obtain:
\begin{subequations}\label{costatequation}
\begin{align}
&\dot {\tilde{\lambda}}_{i}(s) =\theta^k\tilde{\lambda}_{i}(s) \frac{f \left| {\tilde{v}^r_{i}(s)} \right|}{D}-\theta^k\rho c^2 \frac{ \tilde{\mu}_{{i+1}}(s)-\tilde{\mu}_{i}(s)}{ \Delta l} ,\quad i = 0,\dots, N-1,\label{costate11}\\
&\dot {\tilde{\mu}}_{i}(s) =\theta^k\frac{2\gamma}{\bar P}\bigg[\frac{\tilde{p}^r_i(s) - P}{\bar P}\bigg]^{2\gamma-1}-\theta^k\frac{1}{\rho}\frac{ \tilde{\lambda}_{{i}}(s)-\tilde{\lambda}_{i-1}(s)}{ \Delta l},\quad i = 1,\dots, N. \label{costate12}
\end{align}
\end{subequations}
For the terminal conditions, we have
\be\label{costate2}
\begin{aligned}
&\tilde{\lambda}_{i}(r)=\tilde{\mu}_{i}(r)=0,\quad i = 0,\dots, N.
\end{aligned}
\ee
For the boundary conditions, we obtain from (\ref{condition2})
\be\label{costate3}
\begin{aligned}
&\tilde{\mu}_{0}(s)=0,~~ \tilde{\lambda}_{N}(s)=-\frac{2 \rho L {\gamma}}{\bar P^{2\gamma}}\big[{\tilde{p}^r_{N}(s) - P}\big]^{2\gamma-1}.
\end{aligned}
\ee
With the terminal conditions (\ref{costate2}) and the boundary conditions (\ref{costate3}), and  the values of $\tilde{p}^r_i(s)$ and $\tilde{v}^r_i(s), i=1,\dots,N,$ obtained through solving (\ref{subequation001}), the approximate values of  $\tilde{\lambda}(l,s)$ and $\tilde{\mu}(l,s)$   can be obtained by solving the system (\ref{costatequation}) backward in time. Moreover, we apply the composite Simpson's rule \cite{gerald2003numerical}  to approximate the objective function~(\ref{obj-rtrans}) and its gradient formulas given by (\ref{g1-dp})-(\ref{g2-dp}). For numerical integration, we divide each time  interval into $M$  subintervals. With the same integers $N$ and $M$, we partition the
space and time interval evenly to obtain the mesh points $l_0, l_1, \ldots, l_N$ and $t_0, t_1, \ldots, t_{rM}$, where the step sizes $h=L/N$ and $\omega=T/(rM)$. Then, we can get the numerical integration of (\ref{obj-rtrans}), (\ref{g1-dp}), (\ref{g0-dp-0}), (\ref{g2-dp}).
\begin{figure}
\centering\includegraphics[scale=1.2]{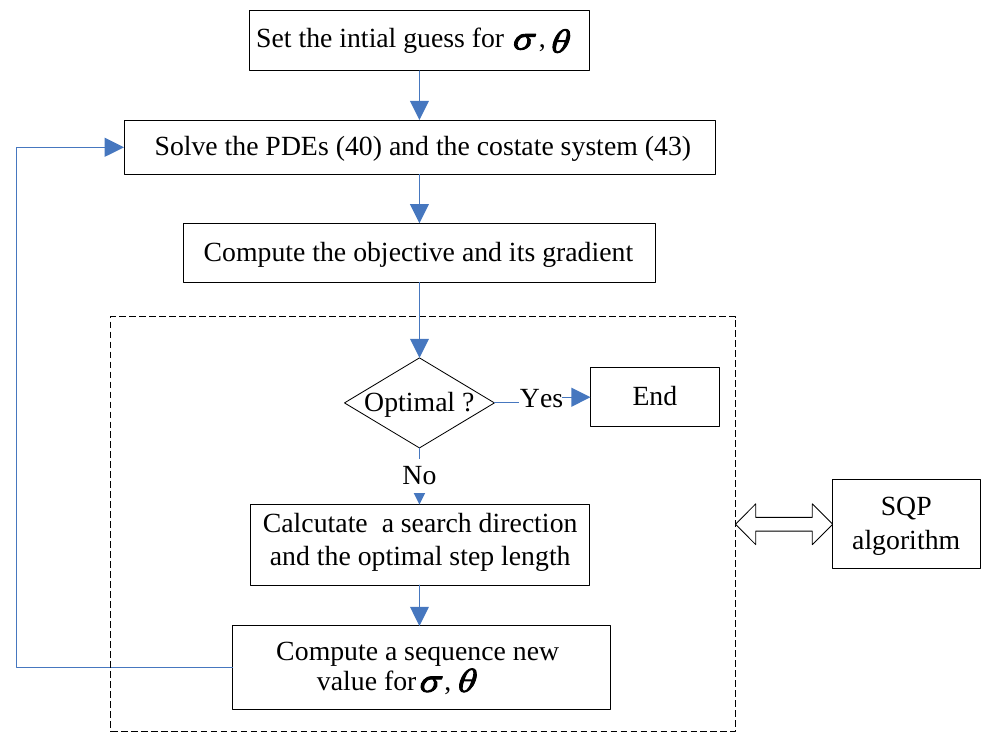}
\caption{Gradient-based optimization framework for solving Problem P$^r_0$ }
\label{process1}
\end{figure}

\subsection{Solving Problem P$^r_0$}
To solve Problem P$^r_0$,  computing the objective function (\ref{obj-rtrans}) and its gradient (\ref{g1-dp})- (\ref{g2-dp})  is the key point.
Since we have  already obtained the values  of $\tilde{p}^r(l,s)$, $\tilde{v}^r(l,s)$, $\tilde{\lambda}(l,s)$ and $\tilde{\mu}(l,s)$, we can  calculate the objective function and its gradient  by applying the numerical integral approximation. Then, we can develop an effective gradient-based optimization technique, such as the SQP method, to solve Problem (P$^r_0$) numerically. The algorithm diagram is shown in Figure \ref{process1}.

Note that Steps 4-5 can be implemented automatically by existing nonlinear optimization solvers, such as FMINICON in MATLAB.

\section{Numerical Simulations}\label{sec:Numerical Simulations}
In this section, we will apply the  proposed computational algorithm to an example to verify the effectiveness of the proposed method in this paper. The pipeline parameters  are taken as: the total pipeline length $L=100$~m, the diameter $D=100$~mm, the flow density $\rho$ =  1000~kg/m$^3$, the  wave speed $c=1200$~m/s, the Darcy-Weisbach friction factor $f=0.03$, $P=2\times10^5$~Pa and  $\bar P = 1\times10^5$ Pa. We also assume that the pipeline fluid flow is initially in the steady state with constant velocity  $\phi_2(l_i)=2~ \textup{m/s},\quad i = 0,\dots,N$. Then the  initial pressure $\phi_2(l)$ is
\begin{equation*}\label{initial pressure}
\phi_2(l_i)=P-\frac{2\rho f }{D}l_i,\quad i = 0,\dots,N.
\end{equation*}
We set $N = 18$ for the spatial discretization of pipeline and choose  $\gamma=2$, $u_{\max}=2 ~ \textup{m/s}$, $T=10$ seconds. Our numerical simulation study was carried out within the MATLAB programming environment (version  R2010b) running on a personal computer with
the following configuration: Intel Core i5-2320 3.00GHz CPU, 4.00GB RAM, 64-bit Windows 7 Operating System.

\begin{table}
  \centering
  \caption{ Optimal control parameters}
  \label{tab2}
  \begin{tabular}{c c c c c c}
   \cmidrule{1-6}
    $k$   & 1    & 2      & 3 & 4    & 5         \\ 
    \cmidrule{1-6}
     ${{\sigma}}_1 ^k$  & $-0.4426$   & $-0.3106$      & $-0.2692$ & $-0.1856$   & $-0.1223 $      \\
         \cmidrule{1-6}
     ${{\sigma}}_2 ^k$  & $2.0000$   & $1.9108$      & $1.8241$ & $1.5972$   & $1.3379 $      \\
     \cmidrule{1-6}
         ${{\theta}} ^k$  & $0.6757$   & $1.1493$      & $0.6207$ & $1.3797$   & $0.6029 $      \\
     \cmidrule{1-6}
      $ k$    &6 & 7    & 8      & 9 & 10    \\
       \cmidrule{1-6}
    ${{\sigma}}_1 ^k$    & $-0.1528$ & $-0.1544$  & $-0.1383$    & $-0.1478$      & $-0.1334$   \\
     \cmidrule{1-6}
    ${{\sigma}}_2 ^k$    & $1.4811$ & $1.4908$  & $1.3866$    & $1.4580$      & $1.3343$   \\
   \cmidrule{1-6}
         ${{\theta}} ^k$  & $1.3668 $  & $0.4275$   & $1.0306$      & $1.0598$ & $1.4169$        \\
     \cmidrule{1-6}

  \end{tabular}
\end{table}

We apply the proposed method to optimize the control sequence $\sigma^k_1, \sigma^k_2, \theta^k, ~k=1,2,\dots,r$. We also set  the number of time segments $r=10$ and the number of subintervals $M=100$.    The  optimal control parameters  are given in Table~{1}.  We compare the optimal control input curves in Figure \ref{controlcurve} obtained by the time-scaling-based method and the time-scaling-free method, respectively. The objective values corresponding to the  time-scaling-based method, time-scaling-free method and constant closure-rate method
are $0.1163$, $0.1512$ and $0.4144$, respectively. Obviously, the constant closure-rate method is worse than the other two methods. {Figure \ref{Comparison} shows} the corresponding pressure changes at the end of the pipeline ($l=L$) associated with valve actuation curves shown in Figure \ref{controlcurve}.    The pressure evolutions along the pipeline according to   both approaches are shown in Figure \ref{time-scaling} and Figure \ref{withouttimescaling}, respectively.  Clearly,  result of the PDE-constraint optimization with the time-scaling approach is better than that without using the time-scaling approach. Using the time-scaling approach, we can change the uniform time interval into a nonuniform time interval. Then, computational optimized time interval will lead to smaller oscillations in the pressure evolution, which is shown in Figure \ref{time-scaling}  comparing to Figure \ref{withouttimescaling}.

\section{Conclusion}\label{sec:conclusion}
In this paper,  we proposed an effective computational method to design  active optimal boundary
control for the Saint-Venant model.  The method of lines is used  to solve the state system and  its costate system. From the numerical simulation, it is observed that result of PDE optimization with time-scaling approach is better than that of PDE optimization without using the time-scaling approach. In the future work, we can apply this method to output command tracking which has been studied in  \cite{rabbani2009flatness,rabbani2010feed} using the differential flatness approach of the simplified Hayami model. 
 For real-time implementation of the proposed control method, we can use feedback control to track the optimal control target if the external perturbation is reasonably small. We can also carry out FPGA-based (Field Programmable Gate Array) implementation for real time optimization instead of software platform combined with model order reduction techniques \cite{yang2012model,xu2013schuster}.

\section*{References}
\bibliographystyle{plain}
\bibliography{RenXu1}
\end{CJK*}
\begin{figure}
\centering\includegraphics[scale=0.85]{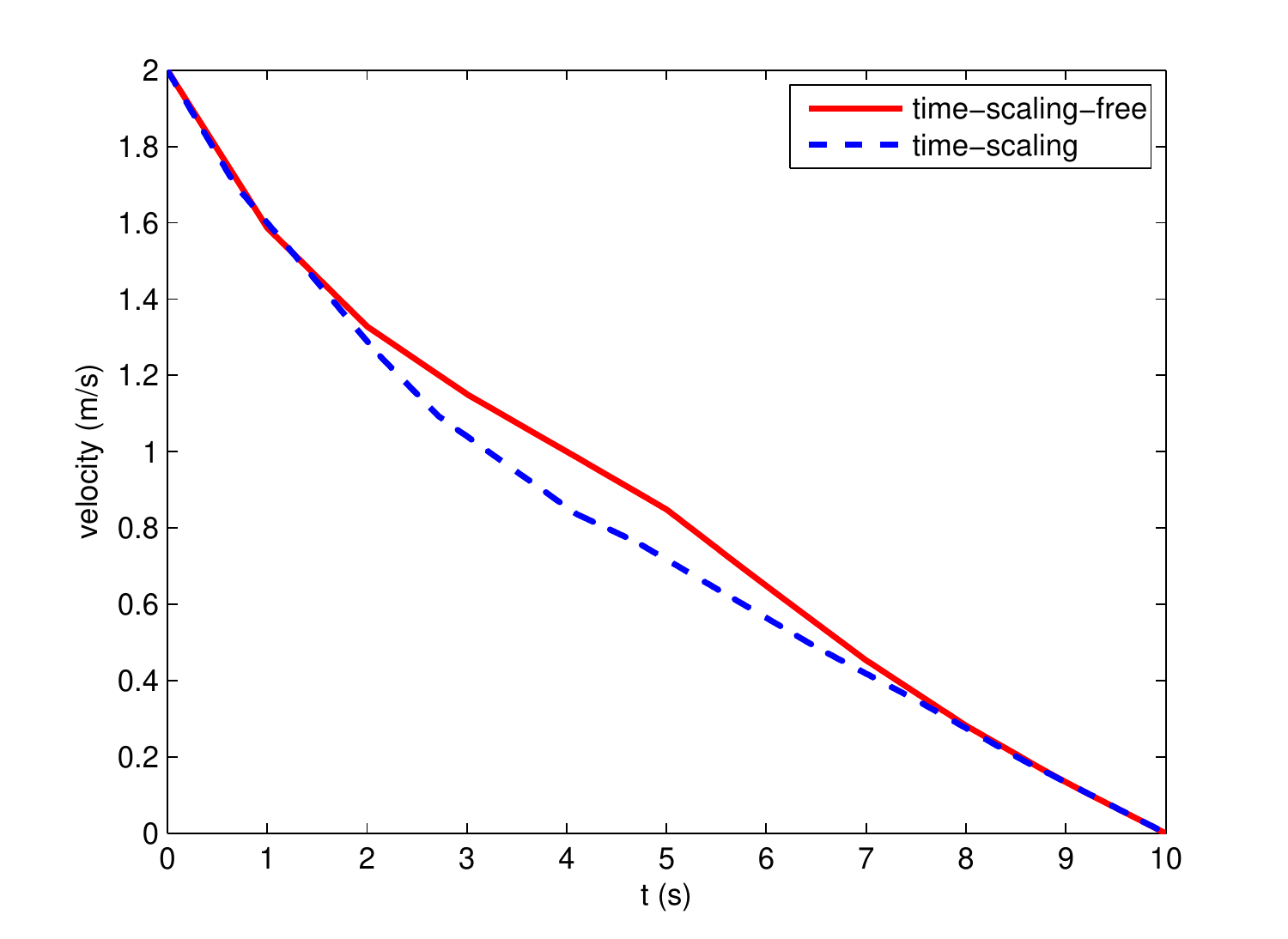}
\caption{Optimal control curves with time scaling approach and without time scaling approach}
\label{controlcurve}
\end{figure}

\begin{figure}
\centering\includegraphics[scale=0.85]{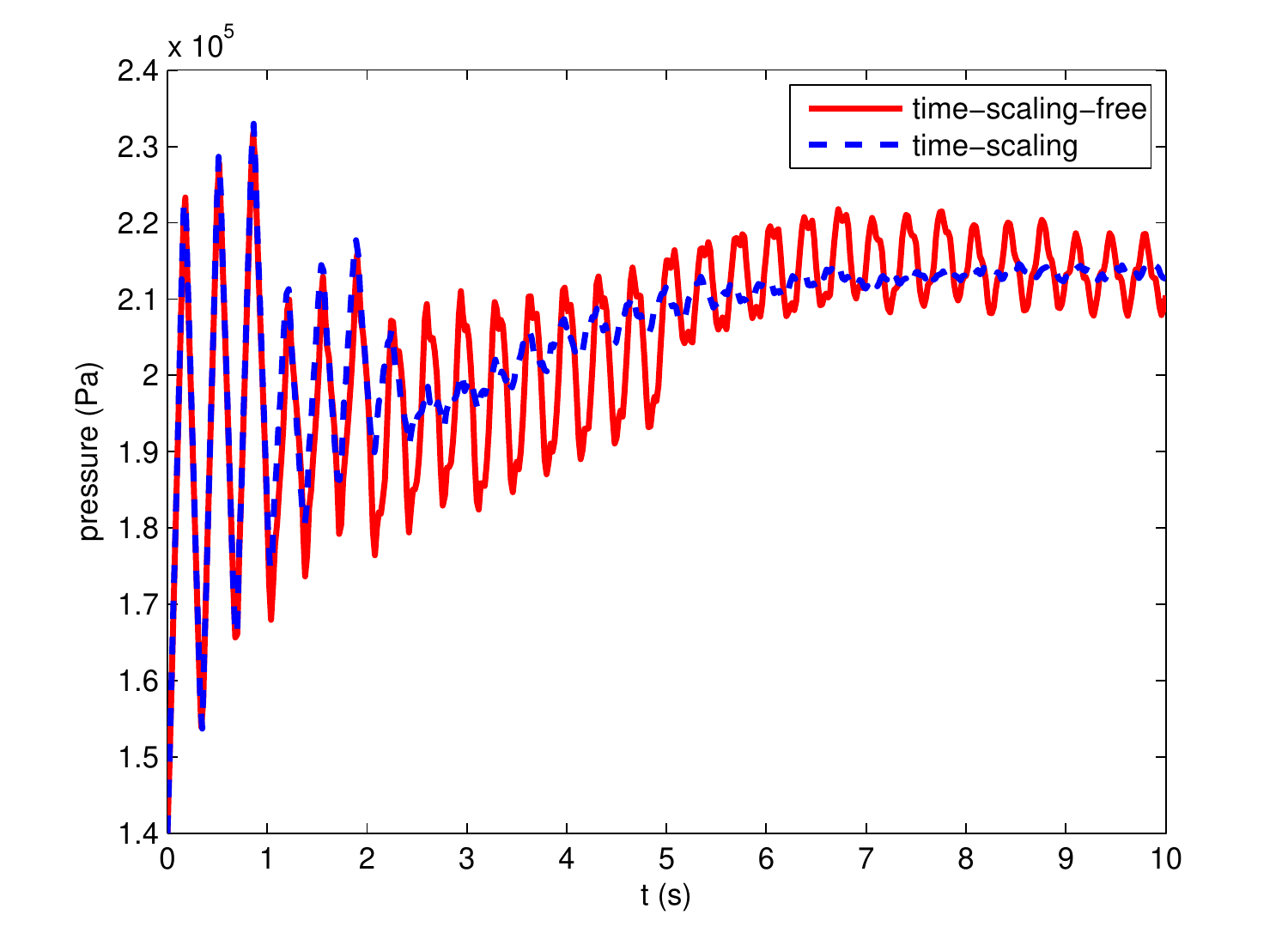}
\caption{Comparison between PDE optimization with time scaling approach and without time scaling approach}
\label{Comparison}
\end{figure}

\begin{figure}
\centering\includegraphics[scale=0.85]{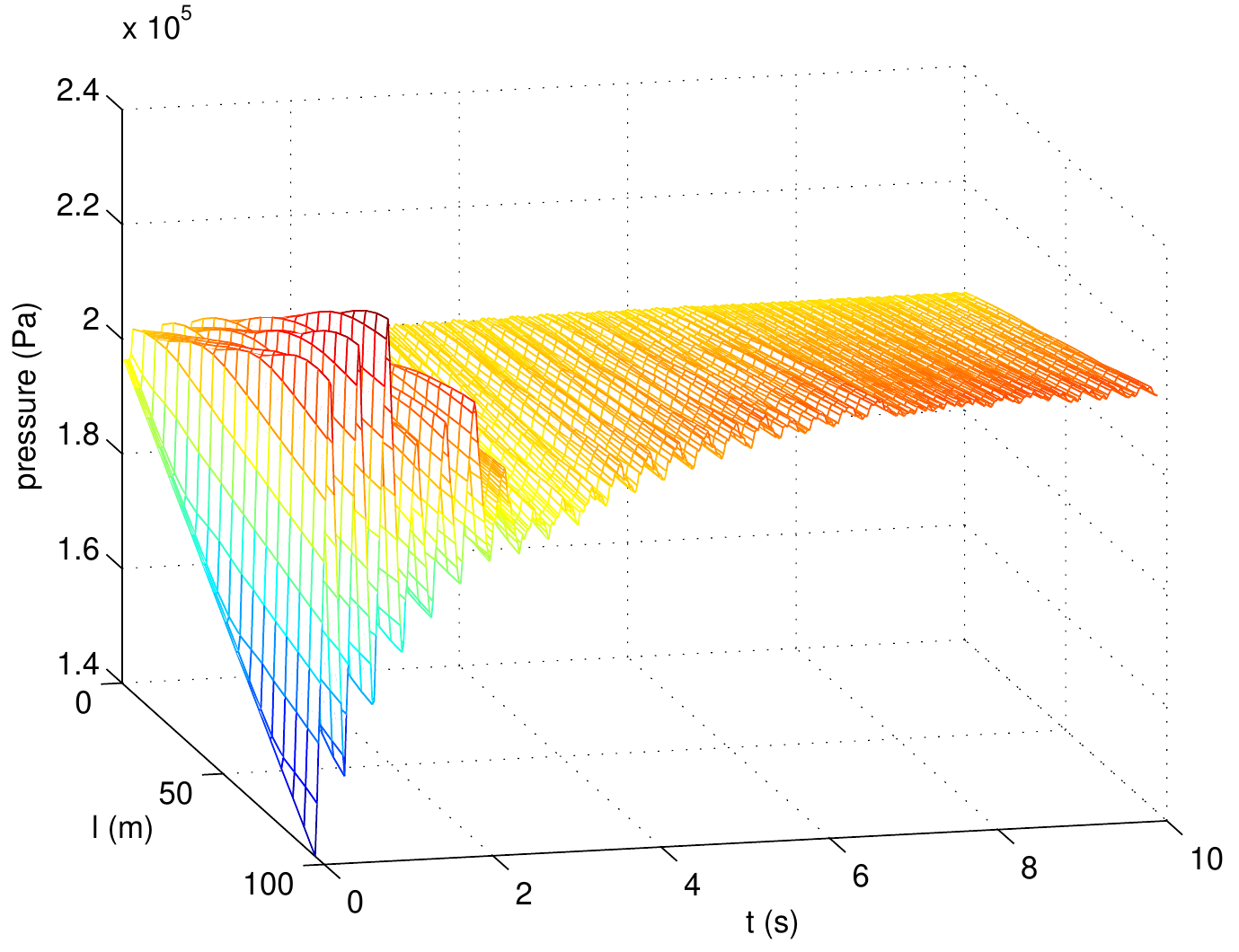}
\caption{PDE optimization with time scaling approach}
\label{time-scaling}
\end{figure}

\begin{figure}
\centering\includegraphics[scale=0.85]{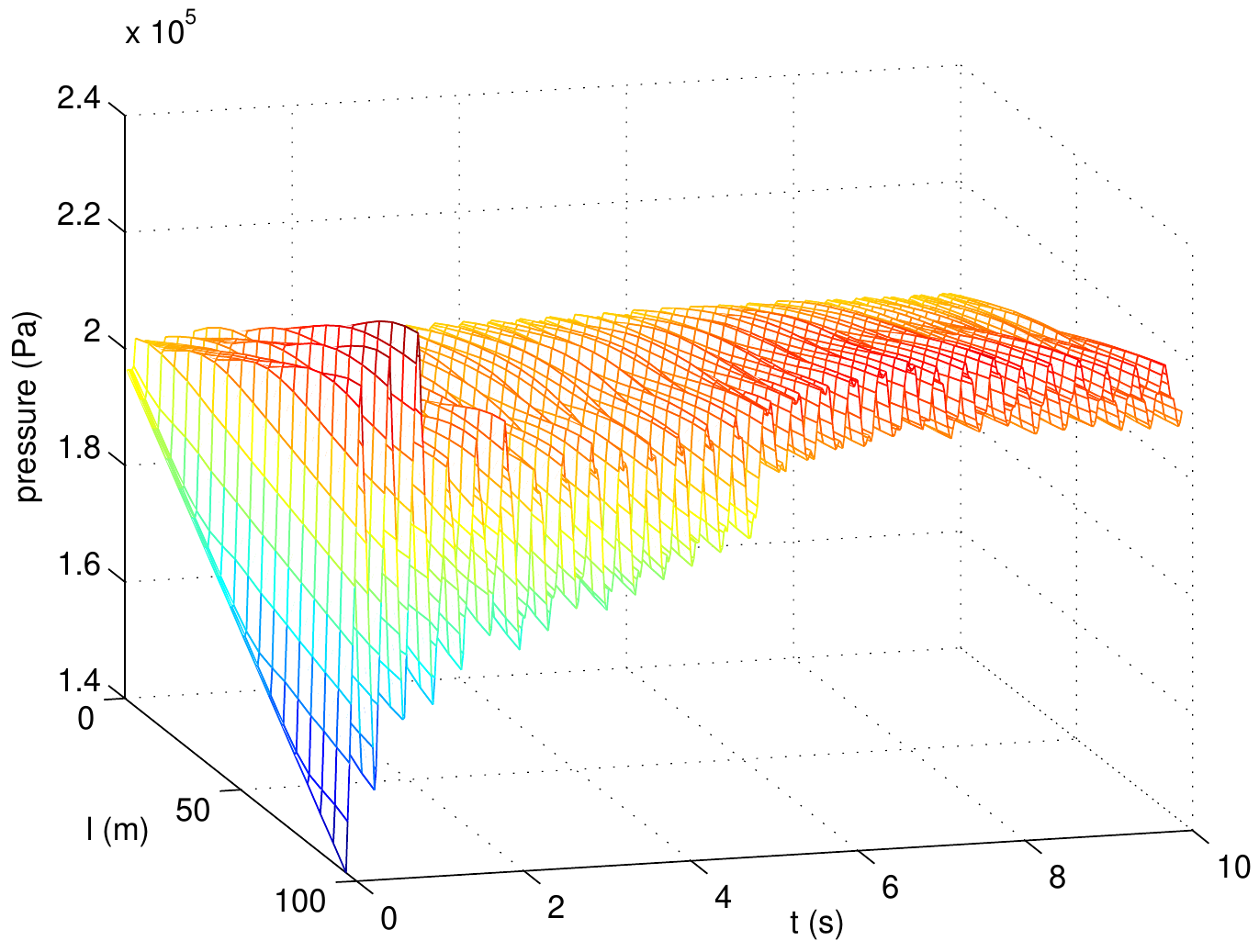}
\caption{PDE optimization without time scaling approach}
\label{withouttimescaling}
\end{figure}
\end{document}